\documentclass[a4paper,cleveref,USenglish,autoref, thm-restate]{lipics-v2021}

\usepackage[utf8]{inputenc}
\usepackage{url}
\usepackage{float}
\usepackage{amsmath}
\usepackage{hyperref}
\usepackage{graphicx}
\usepackage{markdown}
\usepackage{amsthm}
\usepackage{appendix}
\usepackage{xcolor}
\usepackage{tikz}
\usetikzlibrary{arrows}
\usetikzlibrary{patterns}
\usetikzlibrary{positioning}
\usepackage{multirow}
\usepackage{algorithm}
\usepackage{algorithmicx}
\usepackage[noend]{algpseudocode}
\usepackage{appendix}
\usepackage{tikz}
\usepackage{xspace}
\usepackage{booktabs}
\usepackage{subcaption}
\usepackage{multicol}
\usepackage{siunitx}

\title{On Orchestrating Parallel Broadcasts for Distributed Ledgers}

\author{Peiyao Sheng\footnote{Both authors contributed equally to the paper.}}{University of Illinois Urbana-Champaign}{psheng2@illinois}{}{}
\author{Chenyuan Wu$^1$}{University of Pennsylvania}{wucy@seas.upenn.edu}{}{}
\author{Dahlia Malkhi}{University of California, Santa Barbara \and Chainlink Labs}{dahliamalkhi@gmail.com}{}{}
\author{Michael K. Reiter}{Duke University \and Chainlink Labs}{michael.reiter@duke.edu}{}{}
\author{Chrysoula Stathakopoulou}{Chainlink Labs}{stchrysa@gmail.com}{}{}
\author{Michael Wei}{VMware}{michael@wei.email}{}{}
\author{Maofan Yin}{University of California, Santa Barbara}{tederminant@gmail.com}{}{}
\authorrunning{Peiyao S et al.}
\Copyright{Peiyao Sheng, Chenyuan Wu, Dahlia Malkhi, Michael K. Reiter, Chrysoula Stathakopoulou, Michael Wei, and Maofan Yin}
\relatedversion{}
\acknowledgements{This work was conducted when all authors were working at Chainlink Labs. We thank Gregory Neven for his help with the discussions and proofreading.}
\nolinenumbers
\begin{document}
\maketitle

\begin{abstract}
This paper introduces and develops the concept of ``ticketing'', through which atomic broadcasts are orchestrated by nodes in a distributed system. The paper studies different ticketing regimes that allow parallelism, yet prevent slow nodes from hampering overall progress. It introduces a hybrid scheme which combines managed and unmanaged ticketing regimes, striking a balance between adaptivity and resilience. The performance evaluation demonstrates how managed and unmanaged ticketing regimes benefit throughput in systems with heterogeneous resources both in static and dynamic scenarios, with the managed ticketing regime performing better among the two as it adapts better.
Finally, it demonstrates how using the hybrid ticketing regime performance can enjoy both the adaptivity of the managed regime and the liveness guarantees of the unmanaged regime.
\end{abstract}

\newcommand{\peiyao}[1]{{\color{red}{ [Peiyao: #1]}}}
\newcommand{\peiyaoadd}[1]{{\color{red}{ #1}}}
\newcommand{\peiyaorm}[1]{{\color{red}{ \st{#1}}}}
\newcommand{\DM}[1]{{\color{magenta}{ [Dahlia: #1]}}}
\newcommand{\MKR}[1]{{\color{orange}{ [Mike: #1]}}}
\newcommand{\TY}[1]{{\color{blue}{ [Ted: #1]}}}
\newcommand{\CY}[1]{{\color{brown}{ [Chenyuan: #1]}}}
\newcommand{\chrysa}[1]{{\color{violet}{ [Chrysa: #1]}}}
\newcommand{\chrysaadd}[1]{{\color{violet}{ #1}}}
\newcommand{\chrysarm}[1]{{\color{violet}{ \st{#1}}}}

\algnewcommand{\LeftComment}[1]{\Statex {\color{teal}\qquad\textbf{\(\triangleright\) #1}}} 
\newcommand{\signed}[1]{\langle #1 \rangle}
\newcommand{\GN}[1]{{\color{teal}{ [Greg: #1]}}}

\newcommand{\sparagraph}[1]{\vspace{1mm}\noindent {\bf #1}\xspace}
\newcommand{\myparagraph}[1]{\paragraph*{#1}}

\algdef{SxnE}[IF]{Upon}{EndUpon}[1]{\textbf{upon event}\ #1\ \algorithmicdo}
\algdef{SxnE}[IF]{UponR}{EndUponR}[1]{\textbf{upon receiving}\ #1\ \algorithmicdo}
\algdef{SxnE}[IF]{Init}{EndInit}{\textbf{Init:}}

\newcommand{\passtrFullName}{unmanaged ticketing regime\xspace}
\newcommand{\passtrShortName}{UTR\xspace}
\newcommand{\actrFullName}{managed ticketing regime\xspace}
\newcommand{\actrShortName}{MTR\xspace}
\newcommand{\ourtrFullName}{hybrid ticketing regime\xspace}
\newcommand{\ourtrFullNameCap}{Hybrid Ticketing Regime\xspace}
\newcommand{\ourtrShortName}{HTR\xspace}
\newcommand{\gswlong}{Global Sliding-Window\xspace}
\newcommand{\gswshort}{GSW\xspace}
\newcommand{\mswlong}{Maximum Sliding-Window\xspace}
\newcommand{\mswshort}{MSW\xspace}
\newcommand{\unwritten}{\textsf{unwritten}\xspace}
\newcommand{\finalized}{\textsf{finalized}\xspace}
\newcommand{\committed}{\textsf{committed}\xspace}
\newcommand{\sysname}{{\sf Ticketing}\xspace}
\newcommand{\pacemaker}{\mathsf{Pacemaker}}
\newcommand{\tckting}{ticketing\xspace}
\newcommand{\ticketmaster}{ticketing-server\xspace}
\newcommand{\ticketmasters}{ticketing-servers\xspace}

\newcommand{\valid}{\mathsf{valid}}
\newcommand{\invalid}{\mathsf{invalid}}
\newcommand{\undefined}{\mathsf{undefined}}
\newcommand{\propose}{\mathsf{Propose}}
\newcommand{\broadcast}{\mathsf{broadcast}}
\newcommand{\abc}{\mathsf{ABC}}
\newcommand{\seal}{\mathsf{Seal}}
\newcommand{\vote}{\mathsf{Vote}}
\newcommand{\commit}{\mathsf{Commit}}
\newcommand{\vt}{\mathsf{verifyTicket}}
\newcommand{\submit}{{\sf submit} }
\newcommand{\chk}{{\sf check} }
\newcommand{\wt}{{\sf write} }
\newcommand{\rd}{{\sf read} }
\newcommand{\bbca}{{\sf BBCA} }
\newcommand{\cons}{{\sf CONS} }
\newcommand{\bcast}{{\sf BBCA-BCAST} }
\newcommand{\bcommit}{{\sf BBCA-COMMIT} }
\newcommand{\bprobe}{{\sf BBCA-PROBE} }
\newcommand{\cparticipate}{{\sf CONS-PARTICIPATE} }
\newcommand{\ccommit}{{\sf CONS-COMMIT} }
\newcommand{\quota}{{\it maxIndividualPending}\xspace }
\newcommand{\mBA}{GSW\xspace} %
\newcommand*\circled[1]{\tikz[baseline=(char.base)]{
            \node[shape=circle,draw,inner sep=.1pt] (char) {\textbf{#1}};}}

\definecolor{LightCyan}{rgb}{0.88,1,1}

\section{ Introduction}

In state machine replication, operations are organized into a totally ordered sequence through an atomic broadcast protocol~\cite{schneider_cs90}.
In this paper, we are interested primarily in Byzantine fault-tolerant (BFT) atomic broadcast protocols in partial-synchrony, for which solutions are myriad, but share certain ingredients.  
In sequential leader-based protocols, one process at a time is designated the leader, and the leader proposes (blocks of) operations/transactions for the next available (i.e., not yet occupied) slot in the sequence.  
After decades of advances in scaling-\textbf{up} leader-based solutions, recent advances increase throughput by scaling-\textbf{out} and allowing parallel proposing to form a \textit{block-DAG}, e.g., SwirlDS~\cite{baird2016swirlds}, Blockmania~\cite{danezis2018blockmania}, Aleph~\cite{gkagol2019aleph}, and Narwhal/Tusk~\cite{danezis2022narwhal}. 
In block-DAG protocols, all nodes propose in parallel for the next group of slots, and then the entire group simultaneously commits. 
In both the sequential-leader and block-DAG paradigms, the slots in which proposed transactions might settle are implicitly left to be the next available slots in the sequence.  

In this paper, we introduce the concept of ``\textit{\tckting}'' to explicitly manage the slots in which proposed transactions might settle.   
We stress that ticketing is separate from the mechanics by which consensus is reached on the operation in a slot.  
Rather, we leverage the protocol by which slot finalization (and commitment) is performed as a black box.  

The goal of \tckting is to orchestrate atomic broadcasts by nodes in a distributed system with parallelism, yet prevent slow nodes from hampering overall progress.
Tickets capture the right to propose transactions to be committed to the totally ordered sequence, and \tckting refers to the method of orchestrating the assignment of privileges for slots in the sequence.
Several properties factor into the success of a \tckting scheme. 
For example, 
we want to allow parallel proposing but throttle fast proposers from depleting system resources.
We want to prevent bad (or even malicious) proposers from slowing down progress. 
And we want all of this to dynamically adapt to changing system conditions. 
These properties are summarized in our problem definition in~\Cref{sec:ticketing}. 

In the BFT literature, when viewed as a \tckting regime, the prevailing approach for orchestrating proposals is through a sequential leader replacement regime. 
The vast majority of protocols employ a uniform regime that rotates leaders in a round-robin manner or via a randomized lottery among all nodes. This requires all honest nodes to participate uniformly.

Adaptive sequential leader replacement methods based on reputation were introduced first in
Carousel~\cite{cohen2022aware} for sequential BFT protocols and later, in the context of block-DAG, in Shoal~\cite{spiegelman2023shoal}, Hammerhead~\cite{tsimos2023hammerhead}, and Mysticeti~\cite{babel2023mysticeti}.
All of the approaches above do not allow more than one-third of the participants to opt-out of participating, and (in the case of block-DAG protocols) do not orchestrate parallel proposing for nodes with varying speeds. 
These approaches may be categorized as \textit{unmanaged} as they are governed by a distributed protocol.

On the other side, in the crash fault-tolerant (CFT) atomic broadcast literature,
distributed shared logs like CORFU~\cite{balakrishnan2012corfu} demonstrated excellent performance with a \textit{managed} \tckting approach.  
In such systems, there is a designated \textit{``\ticketmaster''} that is responsible for assigning proposers to slots.
The \ticketmaster orchestrates proposing but is decoupled from the consensus protocol that finalizes slots.  
Thus, the \ticketmaster primarily serves as a performance enhancement tool, whereas replication via the consensus protocol guarantees safety. 
A similar approach was recently explored for BFT settings in BBCA-Ledger~\cite{stathakopoulou2023bbca}. 

The benefit of this approach is that it can drive latency down to the limit by allowing parallel broadcasts, yet commits can happen as soon as they are delivered.
More concretely, it allows broadcasts to become \textit{finalized} out-of-order, because they have pre-designated slots.
A finalized slot becomes \textit{committed} when the slot preceding it is committed. 
Under good conditions, this can happen instantaneously.  
In order to address \textit{``holes,''} which might be left in the sequence by bad ticket-holders and prevent higher slots from committing, each slot may be finalized by consensus with a special $\bot$ value after an expiration period.

Managed \tckting embodies several desirable properties, including  
(i) automatically adapting to faulty/slow nodes, (ii) supporting parallel proposing, and (iii) permitting participants to opt-out from proposing.

One seeming drawback of managed \tckting would be introducing a centralized bottleneck, namely the \ticketmaster. 
Surprisingly, we demonstrate in~\Cref{sec:eval} that despite this, managed \tckting has excellent performance because it prevents other sources of slowness. 
Additionally, we prove in~\Cref{sec:protocol} (\Cref{th:utilization}) that under crash-failures only, after a bounded ``warm-up'' segment, managed \tckting prevents any holes from forming in the sequence. 

The second drawback, unique to the Byzantine setting, is the threat of a bad \ticketmaster. 
To tackle this problem, we introduce a dual managed/unmanaged regime called \ourtrFullNameCap (\ourtrShortName), a flexible ticketing regime that transitions between the two without sacrificing consistency. 
\ourtrFullNameCap thus strikes a balance between adaptiveness and resilience. 

The performance evaluation demonstrates how managed and unmanaged ticketing regimes benefit throughput in systems with heterogeneous resources both in static and dynamic scenarios, with the managed ticketing regime performing better of the two as it adapts more effectively.
Finally it demonstrates how using the hybrid ticketing regime performance can enjoy both the adaptivity of the managed regime and the liveness guarantees of the unmanaged regime.

\section{System Overview}
\label{sec:overview}
At a high level, this work tackles the classic problem of \textit{log replication} in permissioned settings; for completeness, the fault model and problem definition are provided in \Cref{sec:model}. 

We consider standard BFT atomic broadcast, but with the additional flexibility of allowing proposers to inject values (blocks) into log slots in parallel, while allowing individual proposals to become committed immediately rather than delaying to commit proposals in batches.
To accomplish this, first, the protocol for individual slots should exhibit a property referred to as \textit{out-of-order finality}, on which we elaborate in \Cref{sec:eject}.
Second, and the principal focus of this paper, injecting proposals into slots should be orchestrated wisely: the goal is to allow parallelism while preventing contention and while adapting to varying workloads and dynamic conditions. This leads us to introduce in \Cref{sec:ticketing} the notion of a \textit{\tckting} regime and formulate a set of desirable \tckting properties.

\subsection{Model} \label{sec:model}
Our system involves a network of $n$ nodes $P= \{0,1,\cdots, n-1\}$. Within this network, up to $f$ nodes may be faulty. Crash faults stall the node permanently while Byzantine nodes act in arbitrarily malicious ways.  Nodes that are neither crashed nor Byzantine are correct.

We assume partially synchronous communication~\cite{dwork1988consensus}, indicating that there exists an unknown Global Synchronization Time (GST), after which the communication delays within the network are bounded by $\Delta$. Each communication channel is authenticated, and each node has a public identity established by Public Key Infrastructure (PKI). The notation $\signed{m}_p$ denotes a message $m$ signed using the public key of node $p\in P$. 

The system's nodes implement atomic broadcast via a replicated log. We denote the data structure maintained by each node as $log$, and $log[sn]$ represents the \textit{slot} in the log with slot number $sn$.
We assume that the replicated log exposes a basic interface with a $\broadcast(sn, b)$ method to propose a value (block in the context of building a blockchain) $b$ for slot $log[sn]$. 

We assume that each slot in the log can be in one of three states: \unwritten, \finalized, or \committed. 
If a slot is finalized, then its contents will not be altered in the future, and content in the slot, if any, is reliably replicated and permanently stored. 
Slot commitment is defined inductively: if a prefix is \finalized, then each slot in the prefix is considered \committed.
A slot that is neither finalized nor committed is \unwritten.
We assume that the log interface further exposes a notification event LOG-COMMIT$(sn, b)$ which is triggered at all nodes once a slot $log[sn]$ gets committed with value $b$.

A replicated log must maintain the following guarantees:

\begin{itemize}
    \item \textbf{Consistency} If two correct nodes attain LOG-COMMIT$(sn, b)$ and LOG-COMMIT$(sn, b')$ for the same slot number $sn$, then $b = b'$.
    \item \textbf{Liveness} %
    Eventually every slot $log[sn]$ attains %
    LOG-COMMIT$(sn, b)$ at all correct nodes, where  $b$ is a block or a special $\bot$ value. %
\end{itemize}

\subsection{Out-of-Order Finality} \label{sec:eject}

To benefit from the ability to inject slot proposals in parallel, slots should be able to finalize \textit{out-of-order}.
That is, each proposal explicitly carries with it a ticket for a slot; the log replication protocol then tries to finalize each slot with a proposal ticketed for it. 
Out-of-order finality allows ``holes'' to be left in the log by bad ticket-holders. Holes prevent the commit progress, but not the finality progress, of subsequent slots. That is, higher slots may become finalized without indirectly finalizing all lower slots, unlike many log-replication protocols (e.g., Raft~\cite{raft}, HotStuff~\cite{hotstuff}) that finalize slots in monotonically increasing order.
In order to prevent holes from preventing higher slots from committing, each slot may be finalized by consensus with a special $\bot$ value after an expiration period.

Whereas any consensus algorithm could be used as a per-slot protocol,
our evaluation focuses on BBCA-Ledger~\cite{stathakopoulou2023bbca}. 
For completeness, briefly BBCA-Ledger implements a single-view regime of PBFT per slot, driven by a leader designated as a ``ticket holder'' for that slot. If there is no observed decision for a certain period, nodes ``eject'' and trigger a fallback consensus, which is akin to a view-change mechanism but can determine only one of two possible outcomes: either the ticket-holder's original proposal or $\bot$.

To guarantee that ejecting does not disrupt liveness, it needs to be synchronized across nodes. 
We assume that there exists a view-synchronization module called $\pacemaker$~\cite{bravo2022making,naor2021cogsworth,cohen2022aware}. $\pacemaker$ manages the starting time and the length of the slot timer for each slot and guarantees a minimum time frame for making progress after GST.  More specifically, nodes can access a local pacemaker module $\Gamma[sn]$ maintaining the following guarantee:

\begin{definition}[Synchronized slot]
    With $\pacemaker$, for all slots $log[sn]$ starting after $GST$, $\Gamma[sn]$ of all correct nodes are active for at least $\Delta_p$ time. We say a slot $log[sn]$ starts after $GST$ if the earliest slot timer for $sn$ of a correct node starts after $ GST$.
    \label{def:sync-slot}
\end{definition}

Here the parameter $\Delta_p$ represents the overlap duration in slots sufficient for correct participants to finalize a block after GST. It is determined by the replicated log protocol and the fallback consensus module. 
    
\subsection{Ticketing: Problem Statement} \label{sec:ticketing}
The goal of this work is to develop an efficient \tckting mechanism for \emph{orchestrating} broadcasts; i.e., assigning to nodes the right to propose to slots in the log, referred to as tickets. %
The \tckting module specifies a local interface $\vt(log, sn, p)\in\{\valid,\invalid,\undefined\}$, which allows nodes to verify locally whether node $p$ is eligible for proposing in slot $log[sn]$, given the current view of $log$. If the return value is $\invalid$, the message will be ignored; when $\undefined$ is returned, the message will be buffered and checked again when log gets updated; otherwise, nodes will further process the proposal in the replicated log protocol. One key property of the interface is that if $\vt(log, sn, p) = \valid$ at some correct node, then $\vt(log, sn, p) \neq \invalid$ at any correct node. This ensures consistency in the eligibility of proposals across correct nodes.

Designing a ticketing scheme that enables good performance under varying conditions and workloads surfaces several desiderata. For instance, the ticketing scheme that assigns slots to all nodes uniformly should perform well in a symmetric network.~\footnote{Our subsequent experiments reveal that, even in a statistically symmetrical network, there are nuanced discrepancies in each node's progress due to system bootstrapping, necessitating a more adaptive design.}  However, in situations where many nodes lack data to propose, such an even allocation could unintentionally waste bandwidth by finalizing empty slots. A more refined strategy would wisely allocate more slots, for example, to nodes with a larger pool of payloads. Beyond just the volume of data to propose, nodes equipped with other resources such as better network capabilities and more advanced computational power should be given a greater number of slots, proportional to their contribution to the log.  %
This property is captured by the principle known as {\it meritocracy}, emphasizing the importance of efficiency and smart resource utilization. 

Furthermore, a good ticketing scheme should be adaptable, ready for both ideal and worst-case scenarios. In an ideal setting where the network is well-connected and all nodes are fault-free, it is best to assign each slot to a single node, eliminating any potential conflicts (referred to as \textit{contention-free allocation}). %
In the face of unstable network conditions or instances of node crashes, however, the design should be resilient enough to minimize the waste of resources caused by failures. To help measure this, we introduce a property called {\it slot utilization}, inspired by the leader utilization proposed by Carousel~\cite{cohen2022aware}, which aims to restrict the number of skipped slots after GST in a crash-only execution. %

Additionally, we take into account the {\it chain quality}~\cite{garay2015bitcoin} of the entire log. This aspect ensures that the portion of log slots proposed by Byzantine nodes after GST remains bounded. The ticketing process specifies a sliding window of pending proposals, capturing the system's inherent parallelism. The design of the window size should aim to enable seamless system operation while preventing potential bottlenecks. These essential properties for a desired \tckting regime are summarized as follows:

\begin{itemize}
    \item {\bf Meritocracy} Nodes that are active %
    and equipped with better resources (e.g., larger payload, advanced network, and computational resources) are favored with a higher number of slots.
    \item {\bf Contention-free allocation} A unique node is assigned to every individual slot.
    \item {\bf Slot utilization} In crash-only executions, after GST, the number of skipped slots is bounded.
    \item {\bf Chain quality} After GST, the proportion of blocks contributed by Byzantine nodes is bounded in the committed chain of correct nodes. %
\end{itemize}

\begin{figure*}
    \centering
    \includegraphics[width=\textwidth]{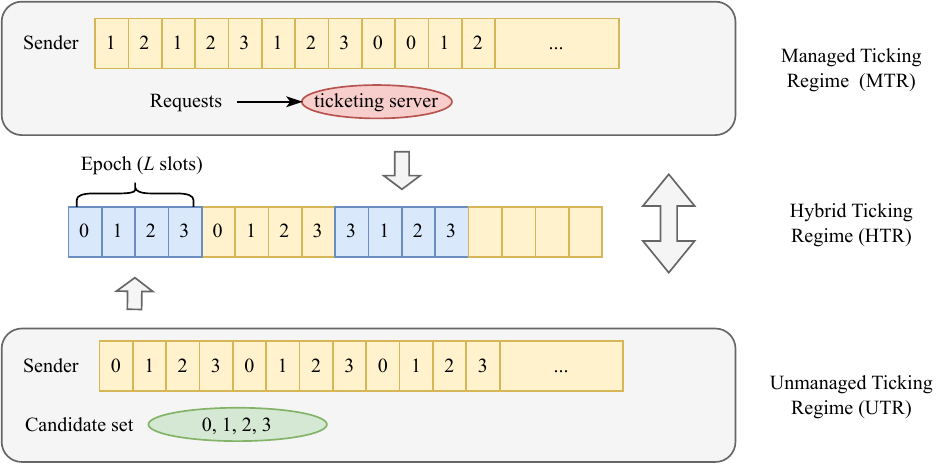}
    \caption{A hybrid managed/unmanaged ticketing regime.} 
    \label{fig:hybrid}
\end{figure*}

\subsection{Technical Approach}
Exploring the dynamics of ticketing in distributed systems, we distinguish between two main patterns: managed and unmanaged. Managed \tckting utilizes a centralized coordinator who listens to ticket requests from potential proposers and assigns slots. Unmanaged \tckting operates in a decentralized manner.

Our analysis compares three specific ticketing regimes, two basic regimes implementing solely managed or unmanaged types, and a dual-mode ticketing regime incorporating both types. On the one hand, the evaluations reveal that the \actrFullName adapts well to dynamic network conditions but is vulnerable to single-node failures. On the other hand, the \passtrFullName offers simplicity and stability in face of faulty environments, yet lacks responsiveness to network changes. Recognizing the limitations inherent in both models, we identify the need for a hybrid paradigm, which integrates the adaptability of the managed approach with the stability of the unmanaged approach, addressing the complexities in system environments and participant behaviors.

The proposed algorithm, called \ourtrFullName, defines a switching mechanism between a \actrFullName (\actrShortName) and a \passtrFullName (\passtrShortName) 
, based on assessments of network stability and performance metrics (Figure~\ref{fig:hybrid}). Specifically, when the network is good and the log is growing without skipped slots, the managed scheme will be adopted to optimize resource allocation. When the network is unstable or a Byzantine \ticketmaster is in place, the log might be stalled, or chain quality is harmed. In this case, the unmanaged scheme will be adopted to resynchronize nodes and bring the system back to a normal pace. This adaptive mechanism ensures efficient resource utilization and good system performance across a range of conditions.

\section{The \ourtrFullNameCap}
\label{sec:protocol}
\subsection{The Protocol}
When designing the \tckting scheme, there are two possible basic paradigms, \textit{managed} and \textit{unmanaged}. 
In a managed \tckting approach, 
a special role is given to one node at a time to actively manage the assignment of the slots.
In an unmanaged \tckting approach, there is a deterministic rule allowing nodes to independently find slot assignments based on their local copy of the log. 
As outlined in \Cref{sec:overview}, we combine both approaches into a hybrid scheme, thus enjoying the agility of a managed approach, coupled with dynamic switching to an unmanaged one for Byzantine resilience.

In Algorithm~\ref{alg:hybrid} we introduce the \ourtrFullName (\ourtrShortName). The protocol groups every $L \geq 2f+1$ slots into an {\it epoch}, such that all slots in each epoch are assigned by the same \tckting scheme.
For each epoch $i$, each node maintains a local candidate set $C[i]$, and a local ticketing scheme $TR[i]$. 
When $TR[i] = -1$, the epoch employs an unmanaged, round-robin scheme, rotating through the nodes in $C[i]$ as eligible proposers. In this case, the $\vt$ function simply verifies whether a proposer for slot $sn$ in epoch $i$ has index $sn \mod L$ in $C[i]$.
When $0\le TR[i] \le n-1$, $TR[i]$ represents the elected \ticketmaster. 
The \ticketmaster accepts requests from nodes and sends signed certificates allocating slots to nodes, which can be verified by the $\vt$ function. 
For completeness, we define that the $\vt(log, sn, *)$ function will return $\undefined$ for all checks to slot $sn$ in epoch $i$ when $TR[i]$ has not been decided. However, correct nodes will never check an unentered epoch and thus will never receive $\undefined$ as the return value. 

To enable high parallelism for data dissemination, our protocol allows $K$ epochs to proceed concurrently. Initially, the first $K$ epochs start simultaneously using an unmanaged round-robin scheme rotating through all nodes (line~\ref{alg:init-begin}-\ref{alg:init-end}). For each subsequent epoch $i>K$, the scheme for the epoch is determined by the outcome of all the slots of epoch $i-K$.  As soon as $TR[i]$ is known, designated proposers can start initiating broadcasts for the epoch. As a result, our protocol allows a node to propose slots at least $(K-1)L$ ahead of the highest committed slot it knows, this bound is denoted as  \textit{\mswlong(\mswshort)}.

Each node subscribes to the LOG-COMMIT event from the replicated log protocol. Upon the commitment of all slots in epoch $i$ (interchangeably, we say epoch $i$ gets committed) (line~\ref{alg:commit-epoch}), %
the protocol reverts the \tckting regime to \passtrShortName (or keeps it, if already using it) if either one of the following two conditions holds: (1) there exists at least one skipped slot in epoch $i$ (i.e., slots committed with a special $\bot$ value), (2) the the number of distinct `active' senders in epoch $i$ is less than $2f+1$. An active sender is a node with at least one proposed slot committed in epoch $i$. If  either (1) or (2) hold, epoch $i+K$ is set to use an unmanaged round-robin scheme (i.e., $TR[i+K] = -1$). Otherwise, a \ticketmaster is selected among the candidate set $C[i+K]$ as explained below (line~\ref{alg:switch-begin}-\ref{alg:switch-end}) through function \textsf{getTicketingServer}.

The candidate set $C[i+K]$ is updated after every unmanaged epoch $i$ (line~\ref{alg:update-cond}). 
Usually, the candidate set for epoch $i+K$ is updated to the set of `active' senders from epoch $i$, containing all senders with at least one proposed slot committed in epoch $i$. %
However, to ensure chain quality, it is imperative to maintain a minimum of $2f+1$ senders. In situations where there are not enough active senders, $C[i+K]$ is reset to the group of all nodes. If the committed epoch $i$ operates with \actrShortName, the protocol keeps the same candidate set to counter potential manipulations by a Byzantine \ticketmaster, who could possibly sideline correct senders to gain unfair election privilege. Figure~\ref{fig:example} illustrates an example of how \ourtrShortName updates different parameters based on execution results.

\begin{algorithm*}
\caption{The \ourtrFullName implementation }
\label{alg:hybrid}
\begin{algorithmic}[1]
    \Init
        \State $n$ \Comment{Number of nodes}
        \State $L$ \Comment{Number of slots per epoch}        
        \State $K$ \Comment{Number of concurrent epochs}
        \State $seed$ \Comment{Common seed}
        \State $C\gets \{\}$ \Comment{Map of candidate set per epoch}  
        \State $TR\gets \{\}$ \Comment{Map of ticketing regime per epoch}
    \EndInit

    \Upon{$INIT()$} \label{alg:init-begin}
        \State $C[1] = C[2] =\cdots = C[K]= [0, 1, \cdots, n-1]$
        \State $TR[1] = TR[2] =\cdots = TR[K]= -1$ \Comment{First $K$ epochs are round-robin epochs}
    \EndUpon\label{alg:init-end}

    \Upon{LOG-COMMIT$(sn, b)$}
        \If{$sn\mod L = 0$} \Comment{All slots in an epoch are committed}\label{alg:commit-epoch}
            \State $i\gets sn / L$ 
            \LeftComment{Update candidate set}
            \If{$TR[i] = -1$} \label{alg:update-cond}
                \State $S\gets \{log[sn].sender\, |\, \forall sn \in [(i-1)L+1, iL], log[sn] \neq \bot\}$  %
                \If{$|S| <2f+1$}
                    \State $C[i+K] = [0, 1, \cdots, n-1]$
                \Else
                    \State $C[i+K] = S$
                \EndIf
            \Else
                \State $C[i+K] = C[i]$
            \EndIf
            \LeftComment{Switch ticketing regime}
            \If{$\exists sn \in [(i-1)L+1, iL], log[sn] = \bot$ or $|S| < 2f+1 \land TR[i]\neq -1 $}\label{alg:switch-begin}
                \State $TR[i+K]\gets -1$ 
            \Else 
                \State $TR[i+K]\gets {\sf getTicketingServer}(i+K)$
            \EndIf\label{alg:switch-end}

        \EndIf   
    \EndUpon

\Function{{\sf getTicketingServer}}{$epoch$}
\label{func:gettm}
    \State $k\gets Hash(seed, epoch) \mod |C[epoch]|$
    \State sort $C[epoch]$ by nodes' public keys in ascendant order
    \State return $C[epoch][k]$
\EndFunction
\end{algorithmic}
\end{algorithm*}

\begin{figure*}
    \centering
    \includegraphics[width=\textwidth]{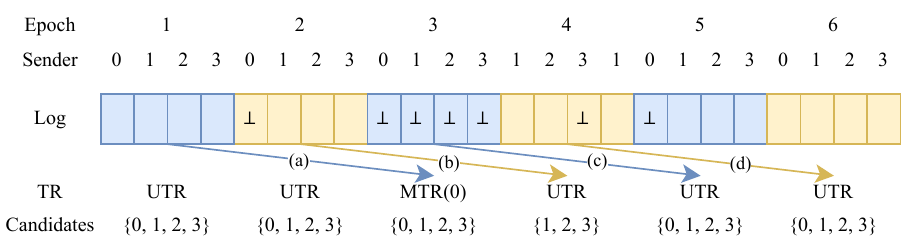}
    \caption{An example with six epochs, $n=4, L=4, K=2$. Log slots with $\bot$ are skipped slots and others are committed with non empty values. The example shows four possible updating rules: (a) epoch 3 uses a \ticketmaster since no slots are skipped in epoch 1; (b) epoch 4 keeps using round-robin since one slot is skipped in epoch 2 and the candidate set is updated to exclude node 0; (c) though all slots in epoch 3 are skipped, the candidate set remains the same since epoch 3 uses the managed ticketing regime; (d) the candidate set is reset to the full group of nodes in epoch 5 since epoch 4 adopts the unmanaged ticketing regime and has only 2 active senders.} 
    \label{fig:example}
\end{figure*}

\subsection{Analysis}

In Section~\ref{sec:ticketing}, we identified the desired properties of an efficient \tckting regime. In this section, we prove that our design satisfies these properties, except for meritocracy, which is demonstrated in Section~\ref{sec:eval}.

We first state the most straightforward property. In epochs with the \passtrFullName and correct \ticketmasters, only one sender will be assigned to each slot, therefore our design is a \textit{contention-free allocation} in these good cases. %

\paragraph*{Slot utilization.} To prove slot utilization, we first prove that the \tckting regime in each epoch is consistent across correct nodes. 

\begin{lemma}[Epoch consistency]  For any epoch $i$ and any two correct nodes $p$ and $q$, let $C_p, C_q$ and $TR_p, TR_q$ denote their local candidate sets and tickecing regimes, then $C_p[i] = C_q[i]$ and $TR_p[i] = TR_q[i]$.
\label{lemma:epoch_consistency}
\end{lemma}

\begin{proof}
    We first demonstrate that for a series of epochs executed sequentially, such as epochs $1, K+1, 2K+1, \cdots$, if correct nodes begin with consistent views of their candidate sets and \tckting schemes, they will maintain these views consistently in all subsequent epochs.
    
    Base case: initially, any two correct nodes $p$ and $q$ have $C_p[1] = C_q[1]$ and start with $TR_p[1] = TR_q[1] = -1$. 

    Inductive Step: for any $i = tK+1 (t\ge 1)$, we will show that $C_p[i] = C_q[i]$ and $TR_p[i] = TR_q[i]$ holds if $TR_p[i-K] = TR_q[i-K]$ and $C_p[i-K] = C_q[i-K]$ holds. 
    
    First we prove the candidate sets are consistent. If $TR_p[i-K] = TR_q[i-K] = -1$, by consistency of the distributed log, all correct nodes recognize the identical set of active senders $S$. As a result, regardless of the number of senders in this set, $C_p[i] = C_q[i]$. If $TR_p[i-K] = TR_q[i-K] \ge 0$, following the given protocol, we have $C_p[i] = C_p[i-K] = C_q[i-K] = C_q[i]$.

    Then we prove the ticketing schemes are consistent. Again by consistency of the distributed log, $TR_p[i] = TR_q[i] = -1$ if there are any skipped slots. Otherwise, correct nodes unanimously elect a  \ticketmaster by invoking the function ${\sf getTicketingServer(i)}$. Given the same seed and the previously demonstrated consistency $C_p[i] = C_q[i]$,  it follows that $TR_p[i] = TR_q[i]$.

By induction, we conclude that the assertion holds for all  $iK +1$.

An identical inductive proof holds for $iK + j$ for every $j = 1,2,\cdots,K$ hence the Lemma holds for all epochs $> 0$.
\end{proof}

To uphold slot utilization, after GST, for any given slot, we must ensure that if there exists a single correct sender, the slot will not be skipped.

\begin{lemma}[Non-skipping epoch] In a crash-only execution, if no nodes have crashed in an epoch $i$ starting after GST, nor are any nodes in $C[i]$ currently crashed, then no slots will be skipped in epoch $i$.
\label{lemma:non-skipping}
\end{lemma}

\begin{proof}
    By Definition \ref{def:sync-slot}, slots in epoch $i$ are synchronized slots, so correct nodes have enough time to participate.  By Lemma \ref{lemma:epoch_consistency}, all correct nodes have the same candidate set. If epoch $i$ is a round-robin epoch, then since all candidates are alive and no nodes have crashed in this epoch, all slots will be finalized. If epoch $i$ uses a \ticketmaster selected from the candidate set, it is alive, and since there is no crash in the current epoch, every slot will be finalized.
\end{proof}

\begin{theorem}[Slot utilization] 
\label{th:utilization}
In a crash-only execution after GST, the number of slots $s$ committed with $\bot$ (called skipped slots) is bounded by $O(fKL)$.
\end{theorem}
\begin{proof}
To simplify notation, we will fix $1 \le k \le K$ and bound the number of skipped slots after GST in the ``$k$-\textit{subsequence}'' of epochs $i$ of the form $k + x\cdot K$, for $x=1, 2, \cdots$. Each such sub-sequence operates independently from others with respect to determining the \tckting regime. We can then multiply the bound we obtain in the end by $K$.

We argue that after GST, within each $k$-subsequence there are most $f+1$ candidate-set resets to a full set. 
If a candidate-set reset happens at, say, epoch $i$, then for all epochs $j = i+xK$, $x=1, 2, ...$ , there are two cases to consider.
Case 1 is that epoch $j$ is unmanaged or it is managed by a nonfaulty \ticketmaster. 
Then there are at least $2f+1$ non-crashed active senders and the candidate-set will not be reset. 
Case 2 is that epoch $j$ is managed by a crashed \ticketmaster $p$. 
Then $j+K$ will be unmanaged. From here on, any higher epoch $j+K+xK$ can switch to a managed regime only if every node in $C[j+K+xK]$ is an active sender and there are no crashed nodes. Hence, each node who has crashed by epoch $i$ cannot become a \ticketmaster (and a fortiori, $p$ cannot become a \ticketmaster again).  Hence, this case can contribute at most $f$ resets.

We now put together the number of skipped slots both cases may contribute. The second case, managed epochs whose \ticketmasters are crashed, may contribute $fL$ skipped slots. 
Additionally, in both bases put together, each crashed node $p$ can contribute at most $2(f+1)(L/n)$ skipped slots,
because after an epoch with (at most $L/n$) skipped slots by $p$, it is removed from the candidate-set. As we showed above, nodes can return to the candidate-set at most $(f+1)$ times. 

The total number of skipped slots per $k$-subsequence 
is therefore bounded by $fL + 2(f+1)(L/n))$, and the total skipped slots by $K \times (fL + 2(f+1)(L/n))$. 
\end{proof}

\paragraph*{Chain quality.} 
Last, to prove \textit{chain quality} after GST, we consider the worst case when Byzantine \ticketmasters are consecutive. We have the following theorem:

\begin{theorem} [Chain Quality]
\label{lemma:correct-block}
Algorithm~\ref{alg:hybrid} satisfies chain quality: at least $(f+1)K$ blocks are proposed by correct nodes in every $2KL$ blocks after GST.
\end{theorem}

\begin{proof}
Since in epochs with the \passtrFullName or with correct \ticketmasters there are at least $2f+1$ senders, by Definition \ref{def:sync-slot} and Lemma \ref{lemma:epoch_consistency}, these senders will successfully commit their slots. Hence, the number of blocks contributed by correct nodes in each such epoch is at least $f+1$. 

In epochs with Byzantine \ticketmasters, it's possible that all blocks are proposed by Byzantine nodes, but this can be detected by counting the number of distinct active senders after the epochs are committed. Once such an epoch $i$ is detected, the upcoming epoch $i+K$ will use \passtrShortName, and it will have at least $f+1$ blocks contributed by correct nodes. %
As a result, in every $2K$ epochs, at most $K$ epochs have Byzantine \ticketmasters, which means at least $(f+1)K$ blocks are contributed by correct nodes in every $2KL$ blocks.  %
\end{proof}

\paragraph*{Remark}
An analysis of the two basic ticketing regimes (managed and unmanaged) is summarized in \Cref{tab:regime-comparison}. The managed ticketing regime naturally supports meritocracy but has an unbounded number of skipped slots in face of a malicious \ticketmaster. As such, it does not ensure slot utilization or maintain chain quality. The unmanaged scheme distributes tickets evenly across all nodes, which does not support meritocracy but is robust against Byzantine behavior, thereby providing good chain quality. Our approach integrates the advantages of both regimes and improves slot utilization by adaptively updating active sender set. To address the chain quality issues associated with managed epochs, we transition them to unmanaged epochs, thereby maintaining overall high chain quality. Additionally, by introducing parallel epochs, our scheme prevents the epochs with undecided ticketing scheme from hindering the progress of consensus.

\begin{table}[]
\centering
\caption{Comparison of different ticketing regimes.}
\label{tab:regime-comparison}
\begin{tabular}{@{}llccc@{}}
\toprule
                            & \multicolumn{1}{c}{Properties}       & Managed      & Unmanaged  & Dual-regime   \\ \midrule
\multirow{3}{*}{Good-case}  & Meritocracy      & $\checkmark$ &              & $\checkmark$  \\
 & Contention-free & \multirow{2}{*}{$\checkmark$} & \multirow{2}{*}{$\checkmark$} & \multirow{2}{*}{$\checkmark$} \\[-3pt] 
 & allocation \\ \midrule
\multirow{2}{*}{Worst-case} & Slot utilization & unbounded    & unbounded    & $O(fKL)$      \\ 
                            & Chain quality    & unbounded    & $O((f+1)/L)$ & $O((f+1)/2L)$ \\ 
                            \bottomrule
\end{tabular}
\end{table}

\section{Evaluation}
\label{sec:eval}

\subsection{Setup}

This paper focuses on \tckting methods for orchestrating parallel proposing of BFT atomic broadcast.
Our experiments vary the number of simultaneously pending broadcasts and who is permitted to propose.
For a practical system, the total number of outstanding broadcasts (that haven't finished) is bounded due to the limited resources, because until they are finalized, pending broadcasts cannot be compacted or checkpointed to secondary storage.
We denote this tuning knob by \textit{\gswlong} ($\gswshort$). The \mswlong ($\mswshort$) defined in Section~\ref{sec:protocol} captures the maximum $\gswshort$ the experiments can set.
The \gswlong mechanism allows nodes
to participate in broadcasts up to a bounded window of size $\gswshort$ beyond their last locally known committed log-slot. 
This self-throttling sliding window allows all nodes to catch up, limits buffering needs, and prevents faster nodes from proceeding too far ahead.

The assignment of slot numbers to nodes is orthogonal to the $\gswshort$ mechanism. 
Indeed, we explore and evaluate three approaches to manage slot allocation.
In the \passtrFullName (\passtrShortName), slots are allocated to permitted nodes in a round-robin rotation. 
In the \actrFullName  (\actrShortName), slots are allocated via an active \ticketmaster. 
In the \ourtrFullName (\ourtrShortName), slot allocation is operated by a hybrid protocol as specified in Algorithm~\ref{alg:hybrid}. We evaluate latency against throughput of these regimes under varying network and failure conditions. 

Our experiments are conducted with a distributed setup on CloudLab. Since this paper studies \tckting rather than transaction dissemination, we use a small transaction size with $2$ bytes payload throughout our evaluation to show the base performance of the broadcasts, and we do not batch multiple transactions per broadcast. In a production system, however, each proposal could batch hundreds of actual transactions and thus scale up throughput (comparable to published results in the literature). We intentionally isolate this away to emphasize the fundamental performance difference between distinct \tckting regimes.

\subsection{Static Heterogeneity}

\begin{table}[]
\centering
\caption{Performance of different ticketing regimes under a heterogeneous setup}
\label{tbl:slowness}
\begin{tabular}{@{}lSSS[table-format=5.0]rrrr@{}}
\toprule
\multirow{2}{*}{Tickets Assignment} &
  \parbox{5em}{\centering Finality} &
  \parbox{5em}{\centering Commit} &
  \parbox{5em}{\centering Throughput} &
  \multicolumn{4}{c}{Proposed by ... (bps)} \\
  & \parbox{5.25em}{\centering latency (ms)} &
  \parbox{5.25em}{\centering latency (ms)} & 
  \parbox{5em}{\centering (bps)} & \multicolumn{1}{c}{0} & \multicolumn{1}{c}{1} & \multicolumn{1}{c}{2} & \multicolumn{1}{c}{3} \\ \midrule
\passtrShortName (node 0)        & 6.3  & 9.3  & 10770 & 10770 & 0 & 0 & 0         \\
\passtrShortName (node 3)        & 9.0  & 14.2 & 4719  & 0 & 0 & 0 & 4719          \\
\passtrShortName (all nodes)                   & 6.1  & 22.8 & 4406  & 1102 & 1102 & 1102 & 1102 \\
\actrShortName (batch=1)   & 3.9  & 3.9  & 2686  & 853 & 661 & 651 & 356     \\
\actrShortName (batch=10)  & 1.6  & 2.0  & 8830  & 3167 & 2863 & 2611 & 2    \\
\actrShortName (batch=100) & 18.2 & 27.3 & 10755 & 3546 & 3790 & 3383 & 10   \\ \bottomrule
\end{tabular}
\end{table}

We first compare different ticketing regimes over a set of nodes that vary in their processing and network speeds. Specifically, nodes $0-$2 are {\tt c6525-25g} instances ($16$-core $3.0$GHz CPU, $25$GB NIC), while the remaining node $3$ is a slower {\tt m510} instance ($8$-core $2.0$GHz CPU, $10$GB NIC). We compare four ticket assignment strategies: \passtrFullName (\passtrShortName) with only node $0$ (a fast node) included in the candidate set and permitted to propose,  
\passtrShortName with only node $3$ (a slow node) permitted, \passtrShortName with all nodes permitted in a round-robin rotation, and \actrFullName (\actrShortName). 
For the first three strategies, $\mBA$ is set to be $100$, %
which is where the throughput saturates without latency impact. For the last strategy, the $\mBA$ bound is not applied, since parallelism is implicitly managed by the active ticketing server and the way each node asks for tickets (i.e., how many and when). Specifically, the active ticketing server batches tickets and distributes a batch to each proposer upon request, where the batch size is a tunable parameter. 
Each proposer requests for the next batch of tickets only when all slots in the previous batch have been finalized. %

Table~\ref{tbl:slowness} summarizes the results, where the last four columns present the number of slots each node has proposed. When the fastest node is known \textit{a priori}, designating it as the single proposer renders best throughput. 
Having a slower node as the single proposer or allowing all nodes to propose in \passtrShortName regime have similar performance: the round robin assignment is bottlenecked by the slowest node in the system.
Although the slots assigned to fast proposers are finalized rapidly, the slots assigned to the slow proposer progress slowly and result in transient holes everywhere in the ledger given round robin's uniform allocation. This further limits the rate at which fast nodes can propose, due to the parallelism constraints imposed by $\mBA$. 
The \actrShortName regime, on the other hand, demonstrates meritocracy by assigning fast nodes more tickets. With ticket batches of $10$, it strikes the sweet spot between latency and throughput. 
We repeated the same experiment on larger scale networks and observed similar trends. 

The results indicate that when nodes have stable and predictable capabilities, by appointing the most capable proposer, the system reaches the best performance. However, without such prior knowledge, we can still approximate the best case with a \actrFullName since it adapts automatically. This seems to leave the system designer with a choice of either sticking to a single capable proposer or paying a little in performance for adaptivity. Our next set of experiments will explore a dynamic setup of the system where it is infeasible to stick to a fast proposer.

\begin{figure*}
\begin{minipage}{0.49\textwidth}
    \centering
    \includegraphics[width=0.9\columnwidth]{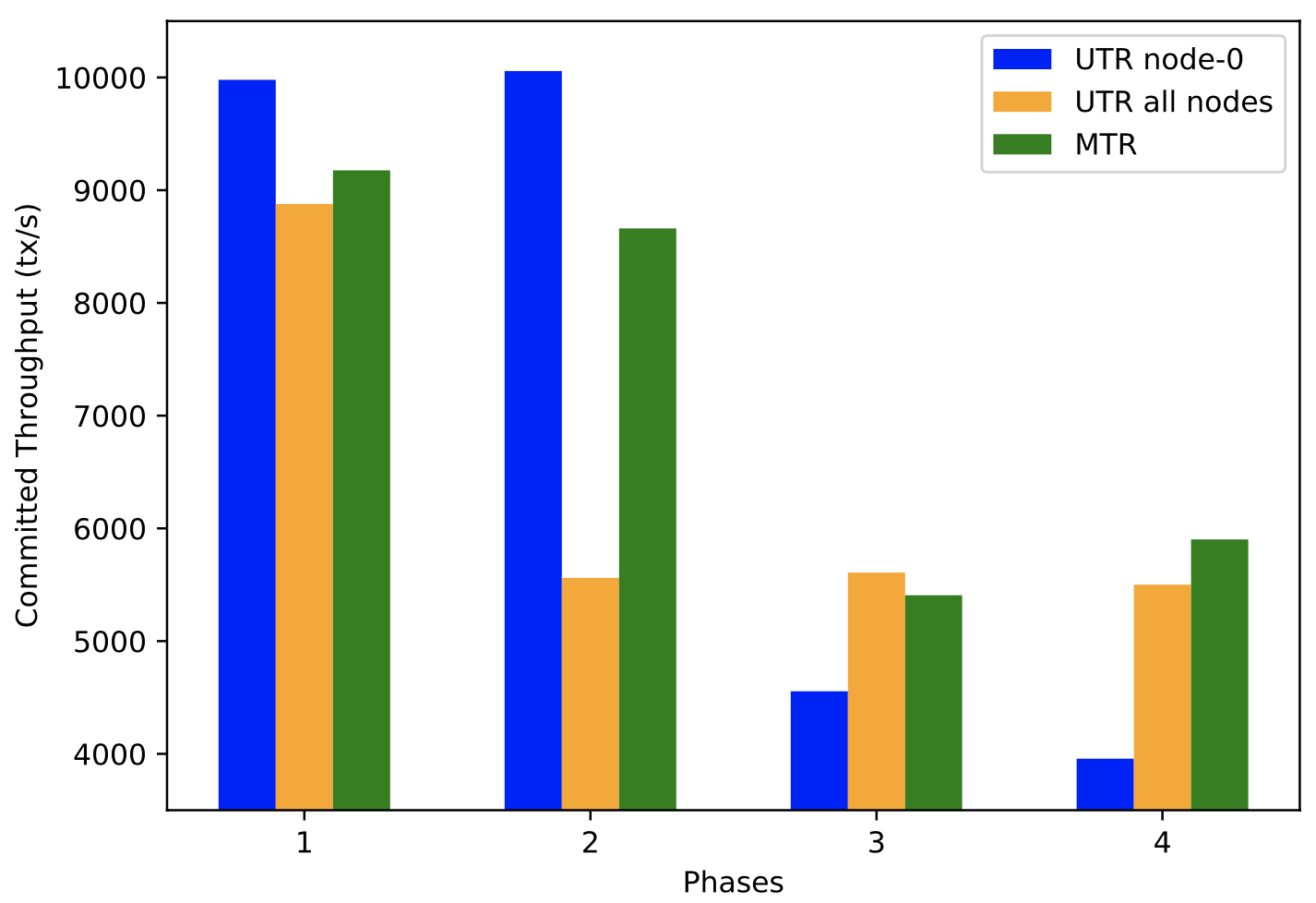}
\end{minipage}
\begin{minipage}{0.49\textwidth}
    \centering
    \includegraphics[width=0.9\columnwidth]{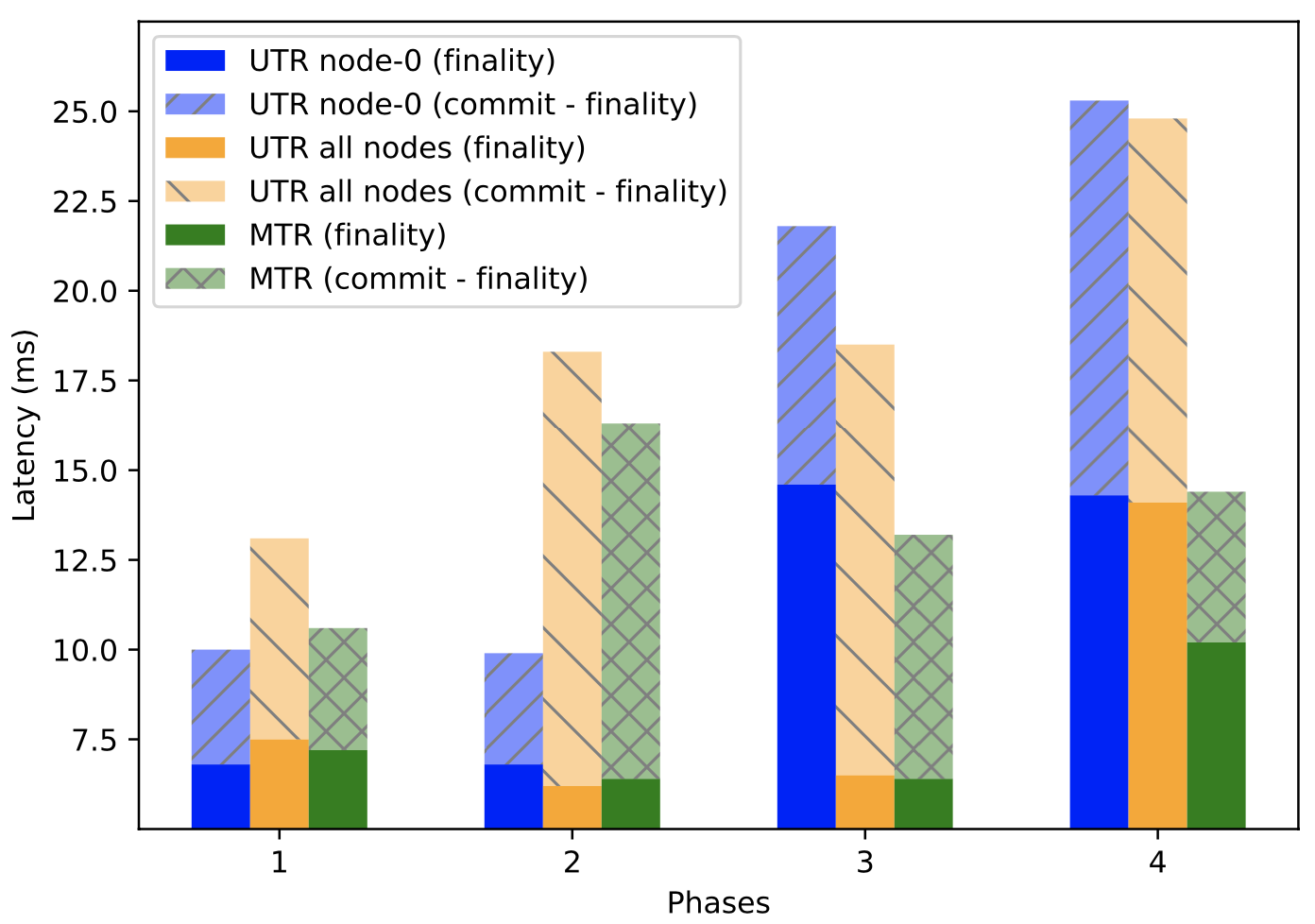}
\end{minipage}
\caption{Throughput and latency of different ticketing regimes under dynamic slowness. Each phase in the experiment lasts $30$ seconds where certain nodes are slowed down. \actrShortName achieves best optimal performance in all phases, demonstrating meritocracy and adaptivity to dynamic conditions. 
}
\label{fig:slow}
\end{figure*}

\subsection{Dynamic Heterogeneity}

We repeat the comparison of different ticketing regimes with heterogeneity, but vary the capabilities of nodes over time. We run four consecutive phases on four {\tt c6525-25g} instances, where each phase lasts for $30$ seconds. In each phase, we slow down certain nodes by idling a half of available CPU cores: in phase $1$, no nodes are slow; in phase $2$, only node $3$ is slow; in phase $3$, only node $0$ is slow; in phase 4, only node $1$ and node $2$ are slow. We compare \passtrShortName with only node $0$ permitted to propose, \passtrShortName with all nodes permitted in a round-robin rotation, and \actrShortName with ticket batches of $10$. We choose this batch size for tickets since it strikes the sweet spot between latency and throughput. 

Figure~\ref{fig:slow} summarizes the performance averaged during each phase. In the latency graph, the solid bar at bottom represents the latency for finality, while the entire bar represents the commit latency. \actrShortName achieves nearly optimal performance in all phases, demonstrating meritocracy and adaptivity to dynamic conditions. Conversely, assigning a single fixed proposer results in lower performance as the capability of the node is not static and thus it does not capture the ``fastest'' node of all time (because there is no such a node). The round robin scheme suffers from poor performance as well. 
In the practical deployment of a system, nodes could run fast and slow at times due to the uneven load imposed by the clients and the handling of different tasks (voting, verification, transaction execution, storage, etc.). With such dynamic heterogeneity, MTR can still adapt much better and mitigates the unnecessary performance loss compared to other approaches.

\begin{figure*}
\begin{minipage}{0.49\textwidth}
    \centering
    \includegraphics[width=0.9\columnwidth]{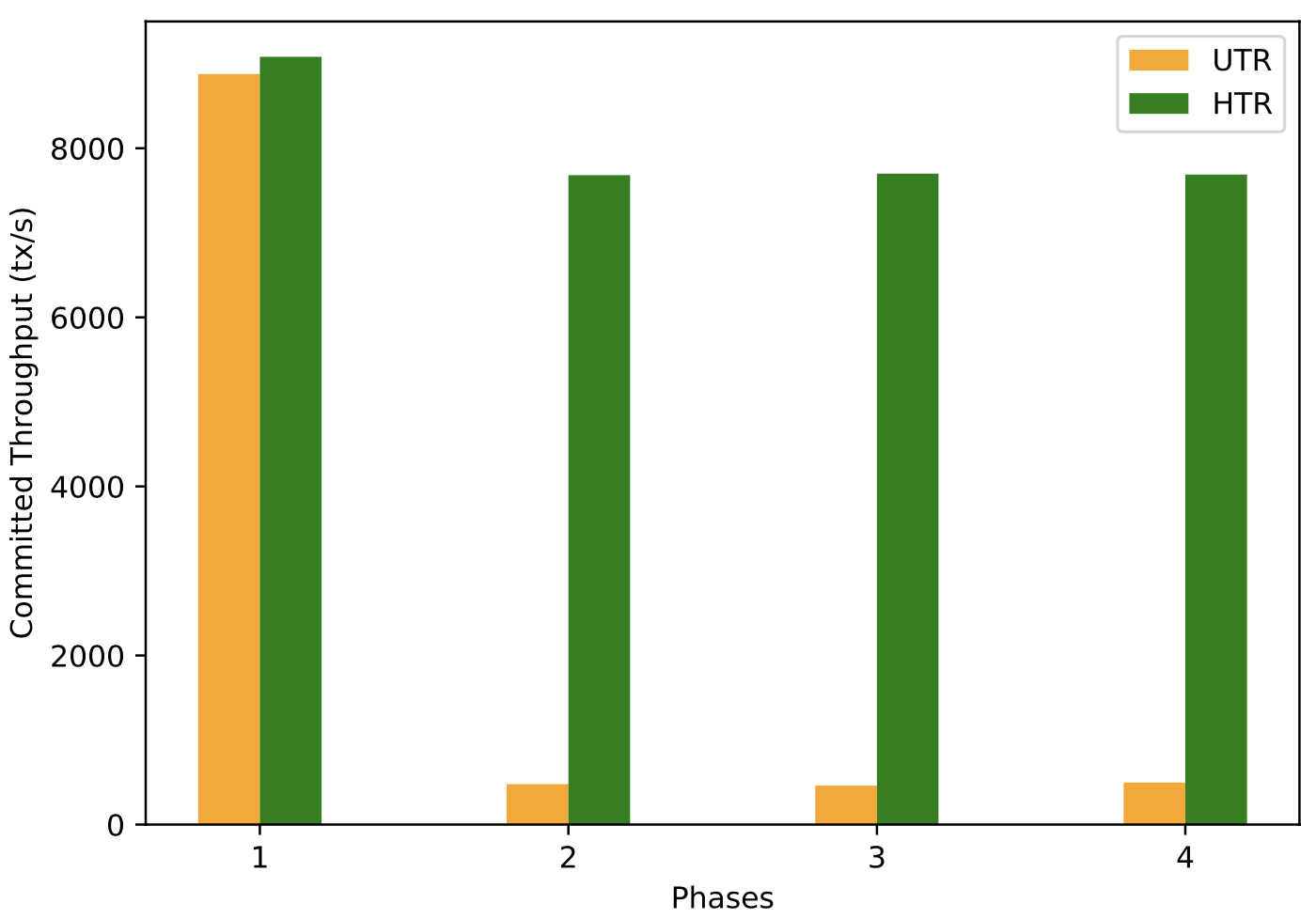}
\end{minipage}
\begin{minipage}{0.49\textwidth}
    \centering
    \includegraphics[width=0.9\columnwidth]{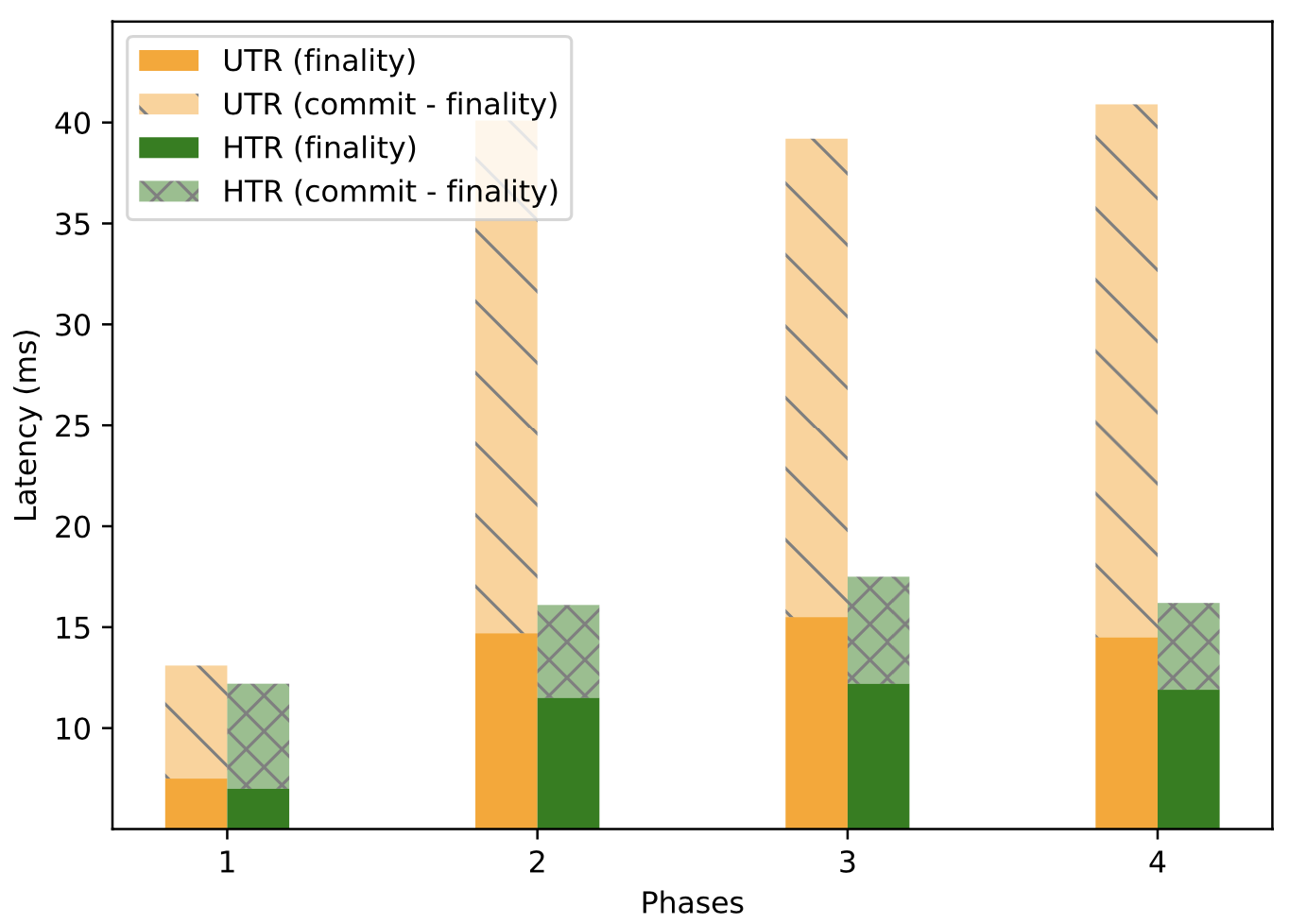}
\end{minipage}
\caption{Throughput and latency of different ticketing regimes under dynamic faults. Each phase in the experiment lasts $30$ seconds where a certain node is faulty and creates skipped slots. \ourtrShortName achieves superior performance in all phases, demonstrating fault resilience. 
}
\label{fig:dual}
\end{figure*}

\subsection{Dual-Mode Regime}

Compared to the \passtrShortName regime with all nodes permitted in round-robin, the main possible drawback for \actrShortName could come from a faulty centralized ticketing server. To address this, we proposed in \Cref{sec:protocol} a dual-mode ticketing regime, and our next experiment evaluates both single and dual-mode regimes with dynamic faults. 

We run four consecutive phases on four {\tt c6525-25g} instances, where each phase lasts for $30$ seconds. In each phase, we vary which node is faulty: in phase $1$, no nodes are faulty; in phase $2$, only node $3$ is faulty; in phase $3$, only node $0$ is faulty; in phase 4, only node $1$ is faulty. The faulty node will not propose slots even when it is assigned with tickets, thus creating skipped slots in the ledger. In all experiments, we use a simulated fallback consensus for simplicity (that is applied to all ticketing designs) and a $10$ms timeout to trigger the fallback consensus. We set the epoch length $L$ to be $50$ and allow $K=2$ concurrent epochs, which effectively sets $\gswshort$ to its maximum value $50$. %

Figure~\ref{fig:dual} summarizes the performance averaged during each phase, where we compare \ourtrShortName versus \passtrShortName with all nodes permitted in a round-robin rotation. Other ticketing regimes suffer from single point failures and are hence not presented in the figure. In the latency graph, the solid bar at bottom represents the latency for finality, while the entire bar represents the commit latency. \ourtrShortName exhibits superior performance in all phases, since the protocol is designed to bound the number of skipped slots. On the contrary, \passtrShortName has unbounded skipped slots, and thus suffers from major performance loss. This means with a dual-mode design, the performance can remain resilient in the case of a faulty ticketing server. Therefore, it is worthwhile to introduce a centralized role to ticketing, given that the faulty server scenario can be mitigated by switching back to a round-round regime and the faulty server is excluded from candidates.

\section{Related Work}
\label{sec:related}

\paragraph*{Ordering layer in shared log.} A shared log is an abstraction that addresses challenges in data consistency and fault tolerance. CORFU~\cite{balakrishnan2012corfu} pioneered the design by separating ordering from replication, introducing a centralized sequencer to manage ordering as a separate layer. The subsequent advance, Scalog~\cite{ding2020scalog}, replaces the centralized sequencer with a replicated counter service using Paxos to improve robustness, and aggregates requests through a tree structure to reduce communication. Recently, FlexLog~\cite{giantsidi2023flexlog} combines the tree structure of sequencer nodes and the single sequencer design in the normal path to further enhance efficiency, and allows for multi-record appends to concurrently append logs.

\paragraph*{Leader election in consensus.} 
Leader-based consensus algorithms (e.g.,~\cite{paxos,raft,pbft,hotstuff,tendermint}) have been widely used to address issues of coordination and agreement among distributed nodes. These protocols usually contain a leader election phase to ensure consistent decision-making on leader rotations and tolerate leader failures in both benign and malicious settings. The simplest leader election scheme, round-robin rotation~\cite{veronese2009spin,streamlet,hotstuff}, inherently guarantees fairness but may continuously elect faulty leaders despite evidence of their misbehavior. 

To address this issue, Ardvaark and follow up works~\cite{clement2008making, veronese2009spin,amir2010prime,aublin2013rbft, stathakopoulou2022state} temporarily eliminate from the leader candidate set  nodes that are suspected to be faulty with a blakclist mechanism.
Another line of work~\cite{miller2016honey, duan2018beat, chen2016algorand, david2018ouroboros} leverages randomized leader election to prevent a succession of faulty (corrupted) leaders.

Recent work~\cite{cohen2022aware} introduces a property called \textit{leader utilization} to bound the number of faulty leaders in crash-only executions after the global stabilization time (GST). This property has been incorporated into some blockchain systems~\cite{spiegelman2023shoal,tsimos2023hammerhead, babel2023mysticeti} to build a reputation system for leaders, enhancing election quality.

Our paper introduces a more flexible scheme which goes one step further than avoiding potentially faulty leaders. 
It allows for faster nodes to broadcast more often, such that networking imbalances have less impact on throughput.

\paragraph*{Orchestrating broadcasts in a DAG.}
Aleph~\cite{gkagol2019aleph} and folllow-up work~\cite{danezis2022narwhal,spiegelman2022bullshark,spiegelman2023shoal,shrestha2024sailfish} separate the consensus logic from broadcasting.
While they adopt a consensus leader election mechanism which falls in one of the categories discussed above, broadcasts are structured in a layered direct acyclic graph (DAG) such that each broadcast block has causal links (edges) to $2f+1$ blocks of the previous layer.
This broadcast layering is incompatible with the out-of-order finality which we examine in this work, as blocks in the DAG are finalized every (few) layers by the consensus leader, therefore, comparison with this line of work is out of scope.
Our paper presents an alternative, more agile approach to the layered ticketing which does not require all nodes to participate in the protocol as block broadcasters. 
Whether the insights of the flexible ticketing studied here can be applied to layered DAG broadcasts is an open problem.

\bibliography{ref}

\end{document}